\documentclass[12pt]{article}
\usepackage[latin1]{inputenc}
\usepackage{lmodern}

\usepackage{amssymb, amsmath, amsthm}
\usepackage[a4paper,top=25mm,bottom=25mm,left=25mm,right=25mm]{geometry}
\usepackage{etex}

\usepackage{authblk} 
\usepackage{pifont}
\usepackage{graphicx}
\usepackage[usenames,dvipsnames,svgnames,table]{xcolor}
\usepackage[figuresright]{rotating}
\usepackage{xtab} 
\usepackage{longtable} 
\usepackage{footnote}
\usepackage[stable]{footmisc}
\usepackage{chngpage} 
\usepackage{pdflscape} 

\usepackage{pgfplots}
\usepackage{setspace}

\makesavenoteenv{tabular}
\usepackage{tabularx}
\usepackage{booktabs}
\usepackage{multirow}
\usepackage{threeparttable}
\usepackage[referable]{threeparttablex} 
\newcolumntype{R}{>{\raggedleft\arraybackslash}X}
\newcolumntype{L}{>{\raggedright\arraybackslash}X}
\newcolumntype{C}{>{\centering\arraybackslash}X}
\newcolumntype{A}{>{\columncolor{gray!25}}C}
\newcolumntype{a}{>{\columncolor{gray!25}}c}

\usepackage{dcolumn} 
\newcolumntype{.}{D{.}{.}{-1}}

\usepackage{tikz}
\usetikzlibrary{arrows}
\usepackage[semicolon]{natbib}
\usepackage{hyperref} 
\usepackage{hyperref}
\hypersetup{
  colorlinks   = true,    
  urlcolor     = blue,    
  linkcolor    = blue,    
  citecolor    = red      
}

\usepackage{microtype}
\usepackage[justification=centerfirst]{caption}

\usepackage[labelformat=simple]{subcaption}

\DeclareCaptionLabelFormat{parenthesis}{(#2)}
\makeatletter
\renewcommand\p@subfigure{\arabic{figure}.}
\makeatother

\newenvironment{customlegend}[1][]{%
	\begingroup
	\csname pgfplots@init@cleared@structures\endcsname
	\pgfplotsset{#1}%
    }{%
	\csname pgfplots@createlegend\endcsname
	\endgroup
    }%
\def\addlegendimage{\csname pgfplots@addlegendimage\endcsname}

\usepackage{enumitem}

\setlist[itemize]{leftmargin=3\parindent}
\setlist[enumerate]{leftmargin=2\parindent}

\theoremstyle{plain}
\newtheorem{assumption}{Assumption}

\newtheorem{observation}{Observation}
\newtheorem{proposition}{Proposition}

\theoremstyle{definition}

\newtheorem{example}{Example}

\theoremstyle{remark}

\def\keywords{\vspace{.5em} 
{\textit{Keywords}:\,\relax%
}}

\def\JEL{\vspace{.5em} 
{\textbf{\emph{JEL} classification number}:\,\relax%
}}

\def\AMS{\vspace{.5em} 
{\textbf{\emph{AMS} classification number}:\,}}

\author{L\'aszl\'o Csat\'o\thanks{~e-mail: laszlo.csato@uni-corvinus.hu} }
\affil{Department of Operations Research and Actuarial Sciences \\ Corvinus University of Budapest \\ MTA-BCE 'Lend\"ulet' Strategic Interactions Research Group \\ Budapest, Hungary}
\title{A mathematical evaluation of vote transfer systems\thanks{~We are grateful to Daniel Bochsler, Tam\'as Solymosi, Bal\'azs Sziklai and Attila Tasn\'adi for their comments and suggestions. \newline
The research was supported by OTKA grant K 111797.}}
\date{\today}

\begin{document}

\maketitle

\begin{abstract}
The paper builds a general model of vote transfer systems on the basis of the current Hungarian electoral rules. It combines single-seat districts and list mandates with three possible compensation rule for 'wasted' votes in constituencies: no compensation (direct vote transfer, DVT), compensation for votes cast for losing party candidates (positive vote transfer, PVT) and compensation for all votes that are not necessary to win the district (negative vote transfer, NVT).

The model is studied in the case of two parties. When the number of votes for the majority party follows a uniform distribution in each district, DVT results in the greatest expected seat share, however, application of PVT, and, especially, NVT increases the probability of winning the election. The trade-off between vote transfer formulas and the number of list mandates reveals that the majority party should use an appropriately calibrated NVT system if it focuses on these two variables.

\end{abstract}

\JEL{D71}

\AMS{91B12}

\section{Introduction}

Mixed-member electoral systems combine direct election of representatives with the aim of a proportional (or, at least more proportional) seat allocation. In this category, most works are devoted to mixed-member proportional (MMP) electoral systems \citep{MassicotteBlais1999, ShugartWattenberg2001}, while other correction mechanisms attracted less attention from academics.

This paper will discuss a close relative of mixed-member proportional systems, the so-called vote transfer system.
They involve two electoral tiers, one consisting of single-seat districts (SSDs) and a second consisting of list seats with a potential compensation for the votes unused in the first tier.
Vote transfer systems may have two advantages over common MMP rules: they seem to be more simple and intuitive (for example, there is no need for overhang seats) and they are immune to some kind of strategic manipulation such as list-splitting \citep{Bochsler2014}.

On the other side, if there are too many compensation seats relative to single-seat districts, a vote-transfer system gives robust incentives for strategic behaviour of the parties and may lead to controversial election outcomes. Furthermore, this design is inevitable in order to provide a fully proportional outcome \citep{Bochsler2015}. Accordingly, vote transfer systems can be suggested as an option just for increasing proportionality while retaining district representation \citep{Raabe2015}.

Probably, it was an important cause of adopting a similar scheme in Hungary during the democratic transition. The country used an especially complex version of positive vote transfer (PVT) system from 1990 to 2010 \citep{Benoit2001, BenoitSchiemann2001}, meaning that parties get a compensation for votes cast for their losing candidates in constituencies. Since there were different (territorial and national) party lists with their own divisor for converting votes into seats, it suffered from the population paradox, namely, some parties might lose by getting more votes or by the opposition obtaining fewer votes \citep{Tasnadi2008}.


After the governing party alliance won a two-thirds (super)majority in the 2010 election, the electoral rules were fundamentally rewritten in 2012: besides a reduction of the number of single-member districts (an issue investigated by \citet{BiroKoczySziklai2015}), the proportional representation pillar was significantly simplified. While vote transfer remained an essential part of the system, its calculation was also modified by the introduction of negative vote transfer (NVT), that is, in addition to compensation due to votes for losing party candidates, the party of the victorious candidate receives the vote difference between the its candidate and the runner-up, too.

This change inspired us to compare different vote transfer formulas. Besides the above mentioned PVT and NVT, the total lack of compensation for 'wasted' SSD votes is also examined under the name of direct vote transfer (DVT). Despite it immediately excludes to achieve clear proportionality, in a number of settings DVT may be functionally equivalent to PVT or NVT with an appropriate choice of the number of list mandates.

For this purpose, a mathematical model of the current Hungarian electoral system is built. All candidates in SSDs are associated with a party, list votes are calculated as the sum of votes in SSDs plus possible compensation according to the vote transfer formula DVT, PVT, or NVT. Share of mandates from SSDs can be chosen arbitrarily.

The model is studied in the case of two parties.
Analytical results can be derived if the number of SSDs won is extremely distorted for either party. They show that NVT is better than PVT, while DVT is the most risky for the majority party.
When the number of votes for the majority party follows a uniform distribution, its expected seat share can also be derived analytically, but its chance to get a majority is examined by a simulation. It turns out that DVT results in the greatest expected seat share, however, application of PVT, and, especially, NVT increases the probability of winning the election. The trade-off between vote transfer formulas and the number of list mandates is investigated, too, revealing that the majority party should use an appropriately calibrated NVT system if it aims these two variables.

The paper is structured as follows. Section~\ref{Sec2} presents the model and Section~\ref{Sec3} discusses its properties in a two-party system. Subsection~\ref{Sec31} deals with the cases where exact derivation is possible, and Subsection~\ref{Sec32} uses a simulation to address further issues.
Main analytical results are formulated in Propositions~\ref{Prop1}-\ref{Prop3}, and a conclusion of simulations is presented in Observation~\ref{Obs1}. 
Finally, Section~\ref{Sec4} summarizes the limitations of the model and possibilities for its extension.

\section{The model} \label{Sec2}

Vote transfer systems allocate parliamentary seats in two tiers. One part of mandates can be obtained in single-member constituencies under the majority (first-past-the-post) rule, while the other tier contains some compensatory seats to approach proportionality.
Theoretically, this system is able to deliver a fully proportional seat allocation, however, it depends not only on the weight of the second tier but on the behaviour of parties and voters \citep{Bochsler2014}. It may also require such a large number of compensatory seats that parties can find rational to try to lose some districts in order to gain a stronger representation overall \citep{Raabe2015}.

Nevertheless, vote transfer systems should not necessarily aim proportionality. They may be used to create a parliament without a too fragmented party structure (a potential failure of truly proportional systems), and, at the same time, avoid the domination of majority.

We formulate a mathematical model of a general vote transfer system, which is able to give any seat allocation between pure proportional and strictly majoritarian electoral outcomes. The main assumptions are as follows: 
\begin{itemize}[label=$\bullet$]
\item
It is a mixed-member electoral system with single-seat districts and list seats. All voter has only one vote.\footnote{~Another interpretation can be that there are two votes on separate ballots for each voter, but vote-splitting is not allowed.} All local candidates are associated with a party. Voters are identified by their party vote.
\item
SSDs are won by the party whose candidate get the most votes. Districts have an equal size.
\item
List seats are allocated according to the proportional rule. The voting weights of the parties can be arbitrary (that is, not only integers) in order to eliminate the effects of the apportionment rule used.
\item
There exists no threshold for the parties to pass in order to be eligible for list seats.
\item
The allocation of list mandates is based on the aggregated number of votes cast for the candidates of each party plus optional correction votes from SSDs. Three different transfer formulas are investigated:
\begin{enumerate}[label=\alph*)]
\item
\emph{Direct vote transfer} (DVT): there are no correction votes;
\item
\emph{Positive vote transfer} (PVT): in addition to the number of list votes under DVT, parties also get the votes of their candidates who lost in SSDs;
\item
\emph{Negative vote transfer} (NVT): in addition to the number of list votes under PVT, parties also get the 'wasted' votes of their victorious candidates, that is, the vote difference between the candidate and the runner-up.
\end{enumerate}
\item
Share of mandates available in SSDs is $\alpha \in \left[ 0,1 \right]$.
\end{itemize}

Example~\ref{Examp1} illustrates this electoral system.

\begin{example} \label{Examp1}
Consider an election between two parties ($A$ and $B$) and two districts:
\begin{center}
\begin{tabularx}{0.9\textwidth}{LCCC} \toprule
          		& Number of voters 	& Votes for party $A$ 	& Votes for party $B$ \\ \midrule
    District 1 	& 100				& 65  					& 35 \\
    District 2 	& 100				& 45  					& 55 \\ \midrule
    National 	& 200				& 110  					& 90 \\ \bottomrule
\end{tabularx}
\end{center}

Each party wins one mandate in the first tier, party $A$ obtains District 1 and party $B$ gains District 2.

Let $\alpha = 0.6$. Then seat allocation according to the three transfer formulas is as follows:
\begin{center}
\begin{tabularx}{\textwidth}{llCC} \toprule
    &  & Party $A$ & Party $B$ \\ \midrule
    & National vote share & $110/200 = 55\%$      & $90/200 = 45\%$ \\ \midrule 
    & Mandate share in SSDs & $1/2 = 50\%$      & $1/2 = 50\%$ \\ \midrule          
    \multirow{3}{*}{DVT}    & Direct list votes & $110$      & $90$ \\
          & List vote share  & $55\%$      & $45\%$ \\
          & \textbf{Seat share} & $\mathbf{52}$\textbf{\%}      & $\mathbf{48}$\textbf{\%} \\ \midrule
    \multirow{4}{*}{PVT}    & Direct list votes & $110$      & $90$ \\
          & Losing list votes & $45$      & $35$ \\
          & List vote share  & $155 / 280 \approx 55.36\%$      & $125 / 280 \approx 44.64\%$ \\
          & \textbf{Seat share} & $\mathbf{\approx 52.14}$\textbf{\%}      & $\mathbf{\approx 47.86}$\textbf{\%} \\ \midrule
    \multirow{5}{*}{NVT}	& Direct list votes & $110$      & $90$ \\
          & Losing list votes & $45$      & $35$ \\
          & Wasted list votes & $65-35 = 30$      & $55-45 = 10$ \\
          & List vote share  & $185 / 320 = 57.8125\%$      & $135 / 320 = 42.1875\%$ \\
          & \textbf{Seat share} & $\mathbf{53.125}$\textbf{\%}      & $\mathbf{46.875}$\textbf{\%} \\ \bottomrule
\end{tabularx}
\end{center}

In the case of DVT, list votes are simply the number of votes across all SSDs.
In the case of PVT, some indirect votes are added to this pool, namely, votes that do not win a seat, which is $35$ for party $B$ in District 1 and $45$ for party $A$ in District 2.
In the case of NVT, 'wasted' votes of the victorious candidates (the difference of votes between the two candidates) in SSDs are also added, which is $30$ for party $A$ in District 1 and $10$ for party $B$ in District 2.

Mandate share in SSDs and list vote shares are aggregated by a weighted sum based on parameter $\alpha$. For instance, party $B$ obtains $0.6 \times 50\% + 0.4 \times 42.1875\% = 46.875\%$ of seats under NVT. 
\end{example}

Note that party $A$ is under-represented in SSDs compared to its national vote share, so it benefits from a smaller $\alpha$ in all cases since list votes somewhat adjusts this disproportion.
Formulas DVT, PVT and NVT gradually approximate proportionality.
However, their 'complexity' also rises: in DVT, list votes are not influenced by the distribution of votes among the constituencies; in PVT, list votes depend on the number of losing votes, that is, which party has won; and in NVT, list votes are affected by the number of votes of the winner and the runner-up.

\section{Study of the model} \label{Sec3}

This section investigates the model above in the case of two parties.

\subsection{Analytical results} \label{Sec31}

Let the vote share of the \emph{majority} party be $x \in \left( 0.5; \, 1 \right)$.
Since $x$ is fixed, the election outcome is deterministic if $\alpha = 0$ (only list votes count). However, voters of the party may be distributed arbitrarily among the districts, affecting the number of SSDs won.
Then seat allocation can be derived analytically in two extreme cases.

According to the first scenario, the majority party wins all districts, which is possible as it has more than 50\% of votes at national level. Then list votes are as follows:
\begin{center}
	\begin{tabularx}{\textwidth}{lCC} \toprule
          					& Majority party 	& Other party \\ \midrule
    Direct list votes 		& $x$     			& $1-x$ \\
    Losing list votes 		& $0$     			& $1-x$ \\
    Wasted list votes 		& $x-(1-x)=2x-1$ 	& $0$ \\ \midrule
    List vote share (DVT)	& $x$      			& $1-x$ \\
    List vote share (PVT) 	& $x / (2-x)$    	& $(2-2x) / (2-x)$ \\
    List vote share (NVT) 	& $(3x-1) / (1+x)$  & $(2-2x) / (1+x)$ \\ \bottomrule
	\end{tabularx}
\end{center}

It implies the following result.

\begin{proposition} \label{Prop1}
Consider a two-party system where the majority party wins all districts.
The majority party's preference order on vote transfer formulas is DVT $\succ$ NVT $\succ$ PVT.
\end{proposition}

\begin{proof}
Note that
\[
x > \frac{x}{1 + (1-x)} = \frac{x}{2-x} \quad \text{and} \quad \frac{3x-1}{1+x} = \frac{x + (2x-1)}{(2-x) + (2x-1)} > \frac{x}{2-x},
\]
so DVT and NVT are more favourable for the majority party than PVT. Furthermore,
\[
x > \frac{3x-1}{1+x} \Leftrightarrow (x-1)^2 > 0,
\]
hence the majority party has the preference order DVT $\succ$ NVT $\succ$ PVT.
\end{proof}

The second extremity is when the majority party (marginally) loses as many districts as possible, that is, the ratio of constituencies lost equals to $2(1-x)$.\footnote{~For the sake of simplicity, it is assumed that $2(1-x)$ times the number of SSDs is an integer.} Then list votes are as follows:
\begin{center}
	\begin{tabularx}{\textwidth}{lCC} \toprule
          					& Majority party 	& Other party \\ \midrule
    Direct list votes 		& $x$     			& $1-x$ \\
    Losing list votes 		& $1-x$     		& $0$ \\
    Wasted list votes 		& $x-(1-x)=2x-1$ 	& $0$ \\ \midrule
    List vote share (DVT) 	& $x$      			& $1-x$ \\
    List vote share (PVT) 	& $1 / (2-x)$   	& $(1-x) / (2-x)$ \\
    List vote share (NVT) 	& $2x / (1+x)$  	& $(1-x) / (1+x)$ \\ \bottomrule
	\end{tabularx}
\end{center}

It implies the following result, analogous to Proposition~\ref{Prop1}.

\begin{proposition} \label{Prop2}
Consider a two-party system where the majority party loses as many districts as possible.
The majority party's preference order on vote transfer formulas is NVT $\succ$ PVT $\succ$ DVT.
\end{proposition}

\begin{proof}
Note that
\[
x < \frac{x + (1-x)}{1 + (1-x)} = \frac{1}{2-x} < \frac{1 + (2x-1)}{2-x + (2x-1)} = \frac{2x}{1 + x},
\]
so the majority party has the preference order NVT $\succ$ PVT $\succ$ DVT.
\end{proof}

Naturally, it has a low probability that any of these extremities occur in practice.
Nevertheless, the analysis yields at least two lessons for the majority party, which may be generalized: DVT results in the greatest variance of seat share, and NVT may be better than PVT.

Finally, a simple stochastic case is analysed, first for the expected vote share of the majority party.

\begin{proposition} \label{Prop3}
Consider a two-party system where the vote share of the majority party in each SSD has a continuous uniform distribution over the interval from $k-h$ to $k+h$ (the expected value is $k$) and $0.5 < k < 0.5+h$.
The majority party's preference order on vote transfer formulas with respect to its expected seat share is DVT $\succ$ NVT $\succ$ PVT.
\end{proposition}

\begin{proof}
Based on the properties of uniform distribution, the following table can be filled out for the majority party:
\begin{center}
	\begin{tabularx}{0.8\textwidth}{lC} \toprule
Probability of winning an SSD:				& $\pi = \dfrac{1}{2} + \dfrac{k-0.5}{2h}$	\\ [0.25cm]
Probability of losing an SSD:				& $1-\pi = \dfrac{1}{2} - \dfrac{k-0.5}{2h}$\\ \midrule
Expected vote share in an SSD won:			& $\beta = \dfrac{0.5+k+h}{2}$				\\ [0.25cm]
Expected vote share in an SSD lost:			& $\gamma = \dfrac{0.5+k-h}{2}$				\\ \midrule
Expected vote difference in an SSD won:		& $\Delta = k+h-0.5$						\\
Expected vote difference in an SSD lost:	& $\nabla = 0.5-k+h$						\\ \bottomrule
	\end{tabularx}
\end{center}
Denote by $\ell_i$ the expected share of the majority party from list votes under vote transfer formula $i \in \{ DVT; PVT; NVT\}$.
Then the majority party's expected seat share is $s_i = \alpha \pi + (1-\alpha) \ell_i$. For a given $k$ and $h$, only $\ell_i$ depends on the vote transfer formula applied in the following way:
\begin{eqnarray*}
\ell_{DVT} & = & k \\ 
\ell_{PVT} & = & \frac{k + (1-\pi) \gamma}{1 + (1-\pi) \gamma + \pi (1-\beta)} \\
\ell_{NVT} & = & \frac{k + (1-\pi) \gamma + \pi \Delta}{1 + (1-\pi) \gamma + \pi (1-\beta) + \pi \Delta + (1-\pi) \nabla}.
\end{eqnarray*}
Despite the complicated expressions, it can be seen that $\ell_{PVT} < \ell_{NVT} < \ell_{DVT}$, illustrated by Figure~\ref{Fig1}.
\end{proof}

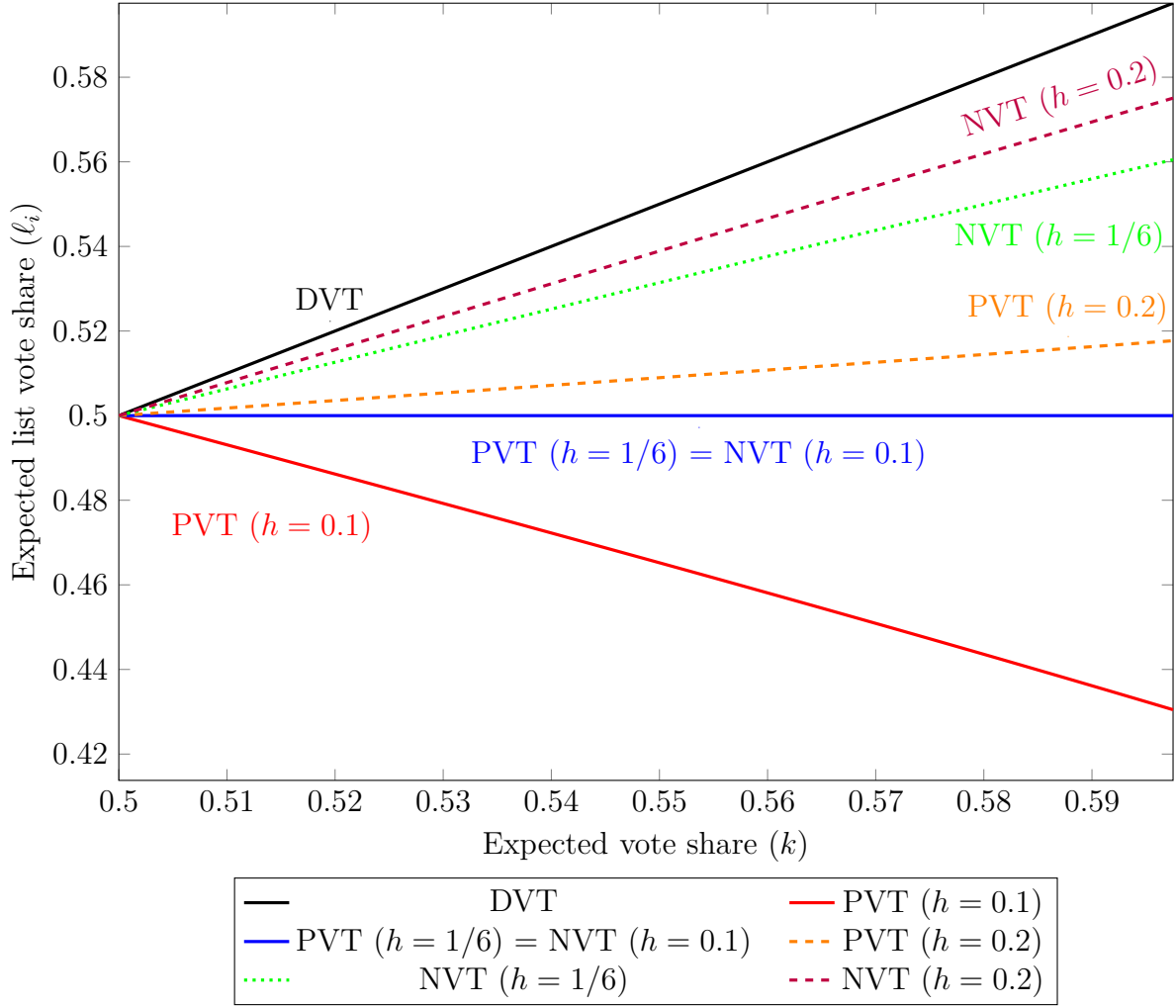
\begin{figure}[tbp]
	\centering
	\caption{Comparison of vote transfer formulas for expected list vote share}
	\label{Fig1}
	
\begin{tikzpicture}
\begin{axis}[width=0.975\textwidth, 
height=0.75\textwidth,
xmin=0.5,
xmax=0.5975,
ymax=0.5975,
xlabel = Expected vote share ($k$),
ylabel = Expected list vote share ($\ell_i$),
legend entries={DVT$\quad$,PVT ($h=0.1$),PVT ($h=1/6$) = NVT ($h=0.1$)$\quad$,PVT ($h=0.2$),NVT ($h=1/6$)$\quad$,NVT ($h=0.2$)},
legend style={at={(0.5,-0.125)},anchor = north,legend columns = 2}
]

\addplot[black,smooth,very thick] coordinates { 
(0.5,0.5)
(0.5025,0.5025)
(0.505,0.505)
(0.5075,0.5075)
(0.51,0.51)
(0.5125,0.5125)
(0.515,0.515)
(0.5175,0.5175)
(0.52,0.52)
(0.5225,0.5225)
(0.525,0.525)
(0.5275,0.5275)
(0.53,0.53)
(0.5325,0.5325)
(0.535,0.535)
(0.5375,0.5375)
(0.54,0.54)
(0.5425,0.5425)
(0.545,0.545)
(0.5475,0.5475)
(0.55,0.55)
(0.5525,0.5525)
(0.555,0.555)
(0.5575,0.5575)
(0.56,0.56)
(0.5625,0.5625)
(0.565,0.565)
(0.5675,0.5675)
(0.57,0.57)
(0.5725,0.5725)
(0.575,0.575)
(0.5775,0.5775)
(0.58,0.58)
(0.5825,0.5825)
(0.585,0.585)
(0.5875,0.5875)
(0.59,0.59)
(0.5925,0.5925)
(0.595,0.595)
(0.5975,0.5975)
}
node [pos=0.2,pin={[pin edge={black}, pin distance=0cm] 90:{DVT}}] {};

\addplot[red,smooth,very thick] coordinates { 
(0.5,0.5)
(0.5025,0.49827582491002)
(0.505,0.496551426847142)
(0.5075,0.494826582742342)
(0.51,0.493101069334253)
(0.5125,0.491374663072776)
(0.515,0.489647140022431)
(0.5175,0.487918275765356)
(0.52,0.486187845303868)
(0.5225,0.484455622962499)
(0.525,0.482721382289417)
(0.5275,0.48098489595713)
(0.53,0.479245935662401)
(0.5325,0.477504272025265)
(0.535,0.475759674487058)
(0.5375,0.474011911207364)
(0.54,0.472260748959779)
(0.5425,0.470505953026393)
(0.545,0.468747287090894)
(0.5475,0.466984513130173)
(0.55,0.465217391304349)
(0.5525,0.46344567984508)
(0.555,0.461669134942069)
(0.5575,0.459887510627631)
(0.56,0.458100558659219)
(0.5625,0.456308028399783)
(0.565,0.454509666695828)
(0.5675,0.452705217753061)
(0.57,0.450894423009472)
(0.5725,0.44907702100573)
(0.575,0.447252747252748)
(0.5775,0.445421334096262)
(0.58,0.443582510578281)
(0.5825,0.44173600229525)
(0.585,0.439881531252764)
(0.5875,0.438018815716659)
(0.59,0.436147570060307)
(0.5925,0.434267504607938)
(0.595,0.432378325473798)
(0.5975,0.430479734396936)
}
node [pos=0.25,pin={[pin edge={white}, pin distance=0cm] 265:{PVT ($h=0.1$)}}] {};

\addplot[blue,smooth,very thick] coordinates { 
(0.5,0.5)
(0.5025,0.5)
(0.505,0.5)
(0.5075,0.5)
(0.51,0.5)
(0.5125,0.5)
(0.515,0.5)
(0.5175,0.5)
(0.52,0.5)
(0.5225,0.5)
(0.525,0.5)
(0.5275,0.5)
(0.53,0.5)
(0.5325,0.5)
(0.535,0.5)
(0.5375,0.5)
(0.54,0.5)
(0.5425,0.5)
(0.545,0.5)
(0.5475,0.5)
(0.55,0.5)
(0.5525,0.5)
(0.555,0.5)
(0.5575,0.5)
(0.56,0.5)
(0.5625,0.5)
(0.565,0.5)
(0.5675,0.5)
(0.57,0.5)
(0.5725,0.5)
(0.575,0.5)
(0.5775,0.5)
(0.58,0.5)
(0.5825,0.5)
(0.585,0.5)
(0.5875,0.5)
(0.59,0.5)
(0.5925,0.5)
(0.595,0.5)
(0.5975,0.5)
}
node [pos=0.55,pin={[pin edge={blue}, pin distance=0cm] 270:{PVT ($h=1/6$) = NVT ($h=0.1$)}}] {};

\addplot[orange,dashed,smooth,very thick] coordinates { 
(0.5,0.5)
(0.5025,0.500446433553946)
(0.505,0.500892897004331)
(0.5075,0.501339420254267)
(0.51,0.501786033220218)
(0.5125,0.502232765838683)
(0.515,0.502679648072886)
(0.5175,0.503126709919487)
(0.52,0.503573981415297)
(0.5225,0.50402149264402)
(0.525,0.504469273743017)
(0.5275,0.50491735491009)
(0.53,0.505365766410302)
(0.5325,0.505814538582818)
(0.535,0.506263701847792)
(0.5375,0.506713286713287)
(0.54,0.507163323782235)
(0.5425,0.507613843759447)
(0.545,0.508064877458667)
(0.5475,0.508516455809679)
(0.55,0.508968609865471)
(0.5525,0.509421370809453)
(0.555,0.509874769962745)
(0.5575,0.510328838791526)
(0.56,0.51078360891445)
(0.5625,0.511239112110143)
(0.565,0.511695380324772)
(0.5675,0.512152445679693)
(0.57,0.512610340479193)
(0.5725,0.51306909721831)
(0.575,0.513528748590755)
(0.5775,0.513989327496925)
(0.58,0.514450867052023)
(0.5825,0.514913400594276)
(0.585,0.515376961693274)
(0.5875,0.515841584158416)
(0.59,0.516307302047472)
(0.5925,0.516774149675284)
(0.595,0.517242161622578)
(0.5975,0.517711372744922)
}
node [pos=0.9,pin={[pin edge={orange,solid}, pin distance=0cm] 90:{PVT ($h=0.2$)}}] {};

\addplot[green,dotted,smooth,very thick] coordinates { 
(0.5,0.5)
(0.5025,0.501578928670582)
(0.505,0.50315774515944)
(0.5075,0.504736337311418)
(0.51,0.50631459302448)
(0.5125,0.507892400276234)
(0.515,0.509469647150425)
(0.5175,0.511046221863366)
(0.52,0.512622012790306)
(0.5225,0.514196908491723)
(0.525,0.515770797739518)
(0.5275,0.517343569543115)
(0.53,0.518915113175427)
(0.5325,0.52048531819871)
(0.535,0.522054074490261)
(0.5375,0.52362127226797)
(0.54,0.525186802115691)
(0.5425,0.526750555008451)
(0.545,0.528312422337452)
(0.5475,0.529872295934878)
(0.55,0.53143006809848)
(0.5525,0.532985631615942)
(0.555,0.534538879788998)
(0.5575,0.536089706457311)
(0.56,0.53763800602208)
(0.5625,0.539183673469387)
(0.565,0.540726604393251)
(0.5675,0.5422666950184)
(0.57,0.543803842222731)
(0.5725,0.54533794355947)
(0.575,0.546868897278999)
(0.5775,0.548396602350356)
(0.58,0.549920958482402)
(0.5825,0.551441866144626)
(0.585,0.552959226587607)
(0.5875,0.554472941863103)
(0.59,0.555982914843765)
(0.5925,0.557489049242477)
(0.595,0.558991249631303)
(0.5975,0.560489421460043)
}
node [pos=0.785,pin={[pin edge={white,solid}, pin distance=-0.25cm] 300:{NVT ($h=1/6$)}}] {};

\addplot[purple,dashed,smooth,very thick] coordinates { 
(0.5,0.5)
(0.5025,0.5019531059267)
(0.505,0.50390609741807)
(0.5075,0.505858860061127)
(0.51,0.50781127948758)
(0.5125,0.509763241396143)
(0.515,0.511714631574837)
(0.5175,0.51366533592324)
(0.52,0.515615240474703)
(0.5225,0.517564231418507)
(0.525,0.51951219512195)
(0.5275,0.521459018152378)
(0.53,0.52340458729911)
(0.5325,0.525348789595296)
(0.535,0.527291512339662)
(0.5375,0.529232643118148)
(0.54,0.531172069825436)
(0.5425,0.533109680686343)
(0.545,0.535045364277091)
(0.5475,0.536979009546423)
(0.55,0.538910505836575)
(0.5525,0.540839742904094)
(0.555,0.542766610940476)
(0.5575,0.544691000592641)
(0.56,0.546612802983219)
(0.5625,0.548531909730647)
(0.565,0.550448212969071)
(0.5675,0.552361605368033)
(0.57,0.554271980151961)
(0.5725,0.556179231119418)
(0.575,0.558083252662147)
(0.5775,0.559983939783863)
(0.58,0.56188118811881)
(0.5825,0.56377489395008)
(0.585,0.565664954227663)
(0.5875,0.567551266586247)
(0.59,0.569433729362751)
(0.5925,0.571312241613583)
(0.595,0.573186703131619)
(0.5975,0.575057014462908)
}
node [pos=0.9,pin={[pin edge={white,solid}, pin distance=-0.1cm, rotate=17.5] 90:{NVT ($h=0.2$)}}] {};
\end{axis}
\end{tikzpicture}
\end{figure}


Uniform distribution can be regarded as a standard case, when no party has an excessive power in any district (or the ratio of party strongholds is approximately proportional to their vote share).
Then Proposition~\ref{Prop3} shows that lack of compensation (DVT) is the most favourable for the majority party. It remains to be seen whether further considerations may modify this conclusion.

\subsection{Simulation results} \label{Sec32}

Sometimes the expected vote share has a modest importance provided that the election is won. Therefore, 50\% can be a natural threshold: a party may be interested in the expected chance of winning the majority of seats.

It is not attempted to derive the exact probability of this even under the restrictive assumption of uniformly distributed votes in SSDs. Therefore, simulations were carried out to get more insight into the effects of vote transfer formulas.

\begin{assumption}
The main assumptions of the calculations are as follows:
\begin{itemize}[label=$\bullet$]
\item
There are $100$ single-seat districts and $m$ mandates allocated on party lists.
\item
The vote share of the majority party in each SSD has a continuous uniform distribution over the interval from $k-h$ to $k+h$ (implying that the expected value is $k$).
\end{itemize}
\end{assumption}

In the simulation, 1 million runs were carried out for every value of $h \in \{ 0.1; 0.2 \}$ and $k$ such that $100k$ is an integer and $0.5 < k < 0.5 + h$. Election outcomes were analysed in the case $m \in \{20; 25; 30; \dots ; 100 \}$ (that is, $\alpha \in \{ 0.8333; 0.8; 0.7692; \dots ; 0.5 \}$).
We have focused on the average seat share of the majority party, and the number of runs when the party obtains a majority (defined as more than 50\% of the seats).

\begin{table}[!ht]
  \centering
  \caption{Simulation results: $\alpha = 0.5$ ($m=100$); $h=0.1$}
  \label{Table1}
    \begin{tabularx}{\textwidth}{l CCCCC} \toprule
    Value of $k$ (expected value) & 0.51  & 0.52  & 0.53  & 0.54  & 0.55 \\ \midrule
    Average vote share & 0.51000 & 0.51999 & 0.53001 & 0.53999 & 0.55000 \\
    Average seat share (DVT) & 0.52999 & 0.55997 & 0.59007 & 0.61994 & 0.65001 \\
    Average seat share (PVT) & 0.52154 & 0.54307 & 0.56467 & 0.58609 & 0.60761 \\
    Average seat share (NVT) & 0.52499 & 0.54998 & 0.57506 & 0.59996 & 0.62501 \\
    Majority: all rules & 864107 & 985589 & 999636 & 999996 & 1000000 \\
    Majority: PVT and NVT & 7516  & 3359  & 135   & 1     & 0 \\
    Majority: DVT and NVT & 0     & 0     & 0     & 0     & 0 \\
    Majority: DVT and PVT & 9     & 0     & 0     & 0     & 0 \\
    Majority: only DVT & 1     & 0     & 0     & 0     & 0 \\
    Majority: only PVT & 0     & 0     & 0     & 0     & 0 \\
    Majority: only NVT & 10928 & 2661  & 103   & 1     & 0 \\
    Minority & 117439 & 8391  & 126   & 2     & 0 \\ \bottomrule
    \end{tabularx}
\end{table}

The case $\alpha = 0.5$ and $h = 0.1$ is detailed in Table~\ref{Table1} for various values of $k$. The first row shows the expected, and the second row shows the average vote share of the majority party. They are approximately equal due to the great number of simulation runs.
In the subsequent three rows, one can see the average seat share according to the three vote transfer formulas, which can be cross-checked with the theoretical value derived in the proof of Proposition~\ref{Prop3}. Due to the large number of runs, they are also very close.

Finally, the last 8 rows highlight the number of cases with respect to winning a majority of seats. In most runs, the majority party wins the election under all vote transfer formulas, but -- due to the stochastic nature of votes -- it may happen that a large number of SSDs, therefore the whole election is lost. Nevertheless, sometimes list votes can compensate it.

There exists no run such that only PVT leads to a favourable election outcome. Furthermore, NVT almost dominates PVT, and PVT almost dominates DVT from this point of view. For example, in the case of $k=0.52$, $14898$ elections are lost under DVT, $11373$ are lost under PVT, and $8782$ are lost under NVT. It is a remarkable difference.

\begin{figure}[!ht]
\centering
\caption{Comparison of vote transfer formulas by simulation}
\label{Fig2}

\begin{subfigure}{0.48\textwidth}
	\centering
	\caption{$\alpha = 0.5$ ($m=100$); $h = 0.1$}
	\label{Fig2a} 
\begin{tikzpicture}
\begin{axis}[width=0.95\textwidth, 
height=0.75\textwidth,
xlabel={Expected vote share ($k$)},
axis y line*=left,
ymin=1,
ymax=50000,
ybar,
ymode=log,
bar width=3pt,
]

\addplot [color=red!50!white,fill] coordinates{ 
(0.51,7515)
(0.52,3359)
(0.53,135)
(0.54,1)
};

\addplot [color=blue!50!white,fill] coordinates{ 
(0.51,18434)
(0.52,6020)
(0.53,238)
(0.54,2)
};
\end{axis}

\begin{axis}[width=0.95\textwidth, 
height=0.75\textwidth,
axis x line*=none,
axis y line*=right,
hide x axis,
ymin=0.5,
ymax=0.66,
]
    
\addplot [color=green,thick] coordinates { 
(0.51,0.52999305694743)
(0.52,0.559969873827318)
(0.53,0.590073432306256)
(0.54,0.619944102916358)
};

\addplot [color=red,thick] coordinates { 
(0.51,0.521544178683437)
(0.52,0.543070315503332)
(0.53,0.564670543695272)
(0.54,0.586088922520421)
};

\addplot [color=blue,thick] coordinates { 
(0.51,0.5249937384947)
(0.52,0.549975282039047)
(0.53,0.5750557034348)
(0.54,0.599957151437912)
};
\end{axis}
\end{tikzpicture}

\end{subfigure}
\begin{subfigure}{0.48\textwidth}
	\centering
	\caption{$\alpha = 0.5$ ($m=100$); $h = 0.2$}
	\label{Fig2b} 
\begin{tikzpicture}
\begin{axis}[width=0.95\textwidth, 
height=0.75\textwidth,
xlabel={Expected vote share ($k$)},
axis y line*=left,
ymin=1,
ymax=50000,
ybar,
ymode=log,
bar width=3pt,
]

\addplot [color=red!50!white,fill] coordinates{ 
(0.51,12454)
(0.52,14686)
(0.53,8878)
(0.54,3186)
(0.55,810)
(0.56,122)
(0.57,17)
};

\addplot [color=blue!50!white,fill] coordinates{ 
(0.51,21257)
(0.52,24365)
(0.53,14199)
(0.54,4867)
(0.55,1128)
(0.56,164)
(0.57,19)
};
\end{axis}

\begin{axis}[width=0.95\textwidth, 
height=0.75\textwidth,
axis x line*=none,
axis y line*=right,
hide x axis,
ymin=0.5,
ymax=0.66,
]
    
\addplot [color=green,thick] coordinates { 
(0.51,0.517525311403929)
(0.52,0.535011147920264)
(0.53,0.552481208006245)
(0.54,0.569979566908149)
(0.55,0.587476959970842)
(0.56,0.60495082921185)
(0.57,0.622523162514183)
};

\addplot [color=red,thick] coordinates { 
(0.51,0.513412901904517)
(0.52,0.526795584630104)
(0.53,0.54017017789165)
(0.54,0.553570992891914)
(0.55,0.566972423344333)
(0.56,0.580360522000538)
(0.57,0.593827783022018)
};

\addplot [color=blue,thick] coordinates { 
(0.51,0.516426688001973)
(0.52,0.532810627590013)
(0.53,0.549176516508179)
(0.54,0.565561335688082)
(0.55,0.581925305947866)
(0.56,0.598248657941766)
(0.57,0.614637714291493)
};
\end{axis}
\end{tikzpicture}
\end{subfigure}

\begin{subfigure}{0.48\textwidth}
	\centering
	\caption{$\alpha = 0.625$ ($m=60$); $h = 0.1$}
	\label{Fig2c}    
\begin{tikzpicture}
\begin{axis}[width=0.95\textwidth, 
height=0.75\textwidth,
xlabel={Expected vote share ($k$)},
axis y line*=left,
ymin=1,
ymax=50000,
ybar,
ymode=log,
bar width=3pt,
]

\addplot [color=red!50!white,fill] coordinates{ 
(0.51,36)
(0.52,636)
(0.53,180)
(0.54,2)
};

\addplot [color=blue!50!white,fill] coordinates{ 
(0.51,2678)
(0.52,3919)
(0.53,329)
(0.54,2)
};
\end{axis}

\begin{axis}[width=0.95\textwidth, 
height=0.75\textwidth,
axis x line*=none,
axis y line*=right,
hide x axis,
ymin=0.5,
ymax=0.66,
]
    
\addplot [color=green,thick] coordinates { 
(0.51,0.534992147710582)
(0.52,0.569964565370461)
(0.53,0.60508838422973)
(0.54,0.639932182187298)
};

\addplot [color=red,thick] coordinates { 
(0.51,0.528655489012568)
(0.52,0.557289896627495)
(0.53,0.586036217771445)
(0.54,0.614540796890309)
};

\addplot [color=blue,thick] coordinates { 
(0.51,0.531242658871036)
(0.52,0.56246862152927)
(0.53,0.593825087576128)
(0.54,0.624941968578465)
};
\end{axis}
\end{tikzpicture}
\end{subfigure}
\begin{subfigure}{0.48\textwidth}
	\centering
	\caption{$\alpha = 0.625$ ($m=60$); $h = 0.2$}
	\label{Fig2d}    
\begin{tikzpicture}
\begin{axis}[width=0.95\textwidth, 
height=0.75\textwidth,
xlabel={Expected vote share ($k$)},
axis y line*=left,
ymin=1,
ymax=50000,
ybar,
ymode=log,
bar width=3pt,
]

\addplot [color=red!50!white,fill] coordinates{ 
(0.51,2910)
(0.52,6389)
(0.53,5483)
(0.54,2041)
(0.55,519)
(0.56,111)
(0.57,13)
};

\addplot [color=blue!50!white,fill] coordinates{ 
(0.51,9645)
(0.52,15545)
(0.53,11152)
(0.54,4163)
(0.55,1121)
(0.56,211)
(0.57,28)
};
\end{axis}

\begin{axis}[width=0.95\textwidth, 
height=0.75\textwidth,
axis x line*=none,
axis y line*=right,
hide x axis,
ymin=0.5,
ymax=0.66,
]
    
\addplot [color=green,thick] coordinates { 
(0.51,0.519359559697398)
(0.52,0.538787472520484)
(0.53,0.558127952374224)
(0.54,0.577510214817471)
(0.55,0.596885605722282)
(0.56,0.616250131118682)
(0.57,0.635659570750487)
};

\addplot [color=red,thick] coordinates { 
(0.51,0.516282917805282)
(0.52,0.532622098809668)
(0.53,0.548891383517129)
(0.54,0.565196325660331)
(0.55,0.581499658171611)
(0.56,0.597797530467797)
(0.57,0.614136949007727)
};

\addplot [color=blue,thick] coordinates { 
(0.51,0.518539411332868)
(0.52,0.537136233490014)
(0.53,0.555649624975631)
(0.54,0.574190284940906)
(0.55,0.592716427787528)
(0.56,0.611218447903139)
(0.57,0.629745403294677)
};
\end{axis}
\end{tikzpicture}
\end{subfigure}

\vspace{-0.5cm}
\begin{center}
\begin{tikzpicture}
        \begin{customlegend}[legend columns=4,legend entries={Average seat share (right scale):$\quad$,DVT$\quad$,PVT$\quad$,NVT}]
        \addlegendimage{empty legend}
        \addlegendimage{color=green,very thick}
        \addlegendimage{color=red,very thick}
        \addlegendimage{color=blue,very thick}
        \end{customlegend}
\end{tikzpicture}
\vspace{0.1cm}
\pgfplotsset{
legend image code/.code={%
\draw[fill] (0cm,-0.1cm) rectangle (0.125cm,0.25cm) (0.25cm,-0.1cm) rectangle (0.375cm,0.175cm);
},
}
\begin{tikzpicture}        
        \begin{customlegend}[legend columns=4,legend entries={{Net advantage over DVT (left log scale, max. $10^6$):$\quad$},PVT$\quad$,NVT}]
        \addlegendimage{empty legend}
        \addlegendimage{color=red!50!white}
        \addlegendimage{color=blue!50!white}   
        \end{customlegend}
\end{tikzpicture}
\end{center}

\end{figure}
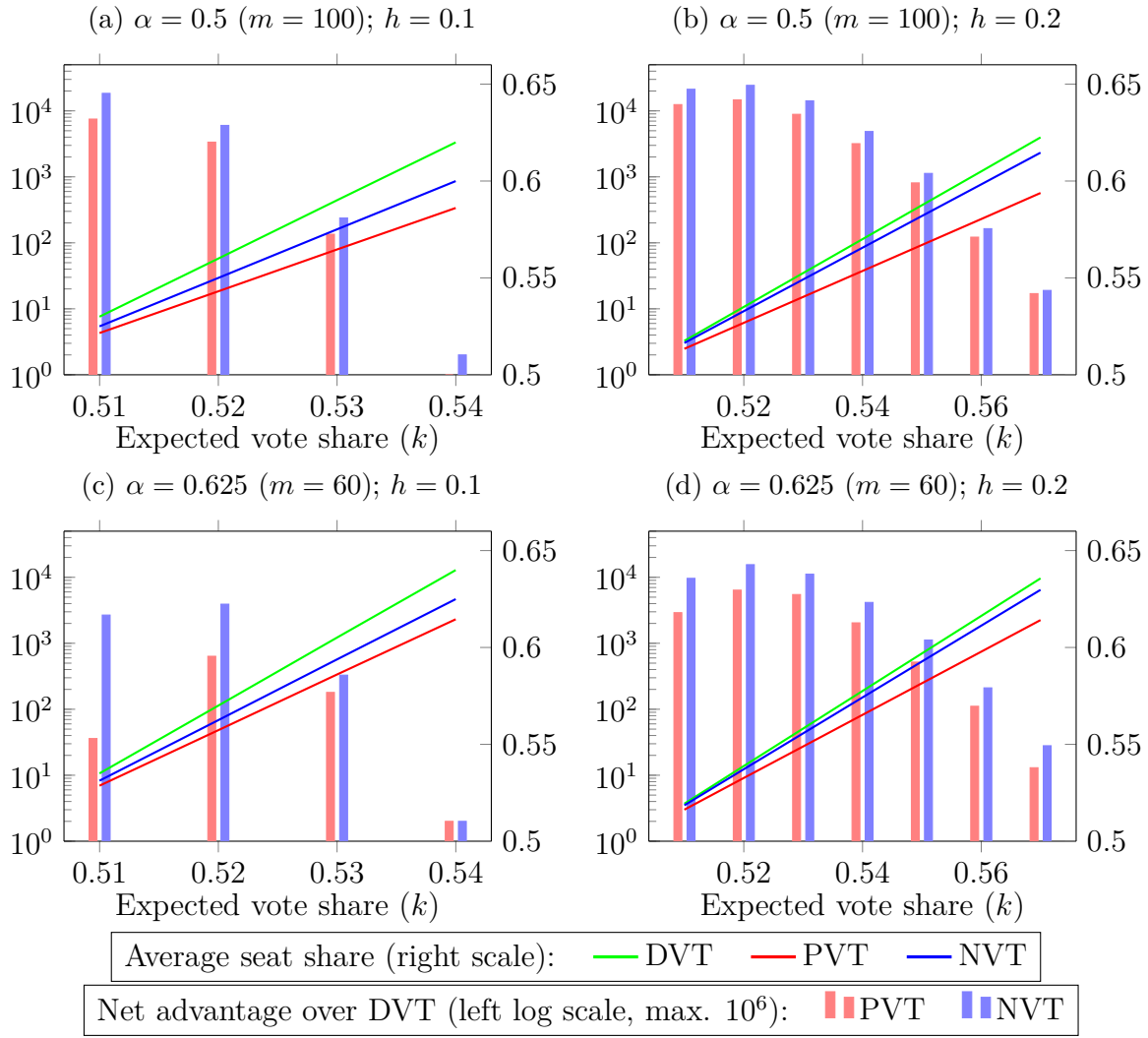


These conclusions are reinforced by further simulations for various values of $h$ and $m$, some of them presented in Figure~\ref{Fig2}. Here the lines (and the right scale) show the average seat share under the three vote transfer formulas, which can be derived from Proposition~\ref{Prop3}. On the other hand, the bars (and the left scale) measure the net advantage over DVT: the number of runs when PVT and NVT results in a majority, but DVT gets a minority, minus the number of runs when DVT results in a majority, but PVT and NVT gets a minority, respectively. For example, in the case of $\alpha=0.5$, $h=0.1$ and $k=0.51$, the net advantage of PVT over DVT is $7515 = 7516 - 1$, while the net advantage of NVT over DVT is $18435 = 7516+10928-17-1$ on the basis of Table~\ref{Table1}. Note that a small difference between the bars for PVT and NVT indicates a significant advantage of NVT due to the logarithmic scale.

On the basis of simulations, the comparison of vote transfer formulas -- discussed by Proposition~\ref{Prop3} from the viewpoint of expected seat share -- can be refined.

\begin{observation} \label{Obs1}
Consider a two-party system where the vote share of the majority party in each SSD has a continuous uniform distribution over the interval from $k-h$ to $k+h$ (the expected value is $k$) and $0.5 < k < 0.5+h$.
The majority party's preference order on vote transfer formulas with respect to its chance to get a majority is NVT $\succ$ PVT $\succ$ DVT.
\end{observation}

Proposition~\ref{Prop3} and Observation~\ref{Obs1} support the conclusions derived analytically for the extreme cases in Propositions~\ref{Prop1} and \ref{Prop2}: application of DVT (despite it leads to the largest expected seat share) is a risky strategy for the majority party because of the diminished probability of winning the election, and NVT seems to dominate PVT from both aspects examined.

Finally, it is possible to outline some trade-offs between the three vote transfer formulas. Assume that the majority party focuses on its chance to obtain a majority of seats. Then, in the case of $h=0.1$ (low variance across SSDs), the majority party is approximately indifferent among NVT with $m=55$ ($\alpha=100/155 \approx 0.6452$), PVT with $m=75$ ($\alpha=100/175 \approx 0.5714$), and DVT with $m=100$ ($\alpha=0.5$) list mandates.

\begin{figure}[!ht]
\centering
\caption{Expected seat share in systems with the same chance to get a majority}
\label{Fig3}

\begin{subfigure}{0.48\textwidth}
	\centering
	\caption{$h = 0.1$}
	\label{Fig3a} 
\begin{tikzpicture}
\begin{axis}[width=0.95\textwidth, 
height=0.75\textwidth,
xmin=0.51,
xmax=0.54,
ymin=0.51,
ymax=0.64,
ymajorgrids,
]
    
\addplot [color=green,thick] coordinates { 
(0.51,0.53)
(0.52,0.56)
(0.53,0.59)
(0.54,0.62)
};

\addplot [color=red,thick] coordinates { 
(0.51,0.525614744000394)
(0.52,0.551223362273086)
(0.53,0.576819686712458)
(0.54,0.602397463839905)
};

\addplot [color=blue,thick] coordinates { 
(0.51,0.532258064516129)
(0.52,0.564516129032258)
(0.53,0.596774193548387)
(0.54,0.629032258064516)
};
\end{axis}
\end{tikzpicture}

\end{subfigure}
\begin{subfigure}{0.48\textwidth}
	\centering
	\caption{$h = 0.2$}
	\label{Fig3b} 
\begin{tikzpicture}
\begin{axis}[width=0.95\textwidth, 
height=0.75\textwidth,
xmin=0.51,
xmax=0.57,
ymin=0.51,
ymax=0.64,
ymajorgrids,
]
    
\addplot [color=green,thick] coordinates { 
(0.51,0.5175)
(0.52,0.535)
(0.53,0.5525)
(0.54,0.57)
(0.55,0.5875)
(0.56,0.605)
(0.57,0.6225)
};

\addplot [color=red,thick] coordinates { 
(0.51,0.515051157094379)
(0.52,0.53010313489227)
(0.53,0.545156757032987)
(0.54,0.560212853049529)
(0.55,0.575272261370916)
(0.56,0.590335832391907)
(0.57,0.60540443163394)
};

\addplot [color=blue,thick] coordinates { 
(0.51,0.518900776592367)
(0.52,0.537798956297475)
(0.53,0.556691950331943)
(0.54,0.57557718606709)
(0.55,0.594452114974269)
(0.56,0.613314220413401)
(0.57,0.632161025215212)
};
\end{axis}
\end{tikzpicture}

\end{subfigure}

\begin{center}
\begin{tikzpicture}
        \begin{customlegend}[legend columns=3,legend entries={DVT ($m=100$)$\quad$,PVT ($m=75$)$\quad$,NVT ($m=55$)}]
        \addlegendimage{color=green,very thick}
        \addlegendimage{color=red,very thick}
        \addlegendimage{color=blue,very thick}
        \end{customlegend}
\end{tikzpicture}
\end{center}

\end{figure}
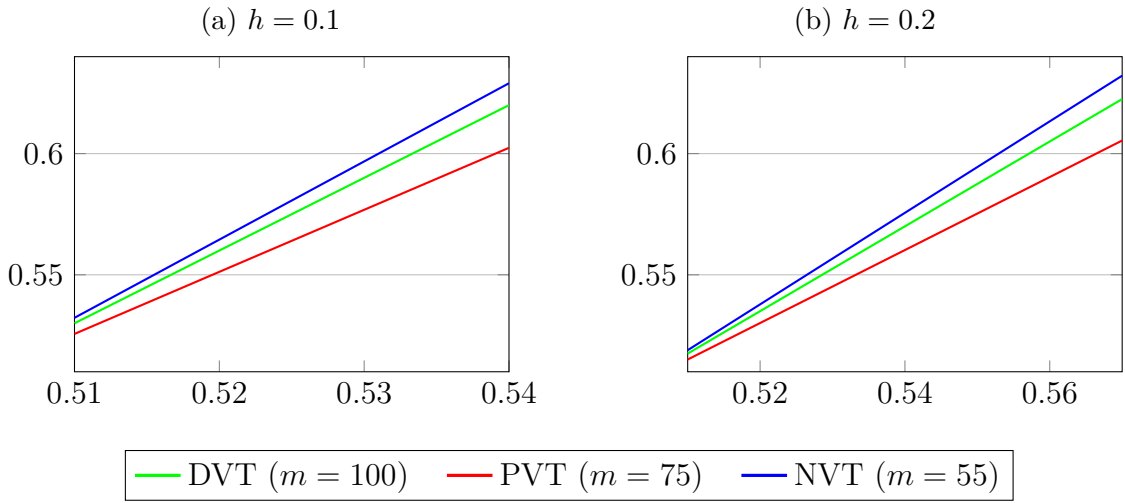


How they perform with respect to expected vote share, calculated in Proposition~\ref{Prop3}? According to Figure~\ref{Fig3a}, NVT with $m=55$ list mandates is the most favourable for the majority party both from the aspects of expected vote share and (marginally) its chance to get a majority. Figure~\ref{Fig3b} shows that the same preferences apply in the case of $h=0.2$ (high variance across SSDs).

To summarize, the majority party should use an appropriately calibrated NVT system in this model, if it concentrates on its average (expected) vote share and its chance to win the election.

\section{Conclusion} \label{Sec4}

The paper has surveyed a model of vote transfer systems in order to compare three different compensation mechanisms. Significant restrictions were necessary in order to to derive mathematically correct implications, including the identification of all candidates with a party, the assumption of no election threshold, and, most worryingly, the focus on a two party competition. Nevertheless, there remain unexplored issues even in this simple framework, for example, increasing the robustness of simulations or a deeper study of the trade-off among the three vote transfer formulas.

Regarding generalization, perhaps the most exciting question is whether these results can be extended for more parties. It is obvious that our conclusions do not change if the two-party system is valid only on the level of SSDs (for example, there are two regional parties which do not compete against each other in constituencies). Since the model does not contain any discontinuity, their list votes and number of mandates in SSDs can be added.

However, if there exist at least one constituency with three parties, the derivation of analytical results seems to be practically impossible. Simulations are also difficult to carry out, because at least two independent random variables should be introduced for each constituency and a lot of possibilities can be found to determine their correlation. For instance, it can occur that the third party has no chance to win a constituency, but it may also be a strong regional party which is a clear favorite in some districts. Therefore, due to the large number of scenarios in an election with at least three parties, this topic is left for future research.


\end{document}